\documentclass[11pt]{article}

\usepackage[utf8]{inputenc}
\usepackage[margin=1in, a4paper]{geometry}
\usepackage[dvipsnames]{xcolor}
\usepackage[T1]{fontenc}
\usepackage{hyperref}
\hypersetup{
	colorlinks=true,
	urlcolor=MidnightBlue,
	linkcolor=MidnightBlue,
	citecolor=MidnightBlue,
	unicode
}
\usepackage{url}
\usepackage{booktabs}
\usepackage{amsfonts}
\usepackage{nicefrac}
\usepackage{microtype} 
\usepackage{xcolor}
\usepackage{complexity}
\usepackage{amsmath}
\usepackage{amsthm}
\usepackage{mathtools}
\usepackage[shortlabels]{enumitem}
\usepackage{xspace}
\usepackage{amssymb}
\usepackage{multicol}
\usepackage[capitalise,nameinlink,noabbrev]{cleveref}
\usepackage{tikz}
\usepackage{thmtools}
\usepackage{doi}

\usetikzlibrary{decorations.pathreplacing}
\usetikzlibrary{arrows.meta}

\newtheorem{theorem}{Theorem}[section]
\newtheorem{lemma}{Lemma}[section]
\newtheorem{claim}{Claim}[section]

\newtheorem{definition}{Definition}

\newclass{\CLS}{CLS}
\newlang{\MinQuadKKT}{MinQuadKKT}
\newlang{\MinmaxIndKKT}{MinmaxIndKKT}
\newcommand{\copyGadget}{\textsc{copy}}
\renewcommand{\epsilon}{\varepsilon}
\newcommand{\eps}{\varepsilon}

\begin{document}

\begin{center}
{\huge The Complexity of Two-Team Polymatrix Games\\[1mm]
with Independent Adversaries}\\[1.3cm] \large
	
\setlength\tabcolsep{1.2em}
\begin{tabular}{ccc}
Alexandros Hollender &
Gilbert Maystre &
Sai Ganesh Nagarajan\\
\small\slshape University of Oxford &
\small\slshape Oracle &
\small\slshape University of Southern Denmark
\end{tabular}

\vspace{6mm}

\vspace{4mm}
\end{center}

\begin{quote}
\noindent\small
{\bf Abstract.}~
Adversarial multiplayer games are an important object of study in multiagent learning. In particular, polymatrix zero-sum games are a multiplayer setting where Nash equilibria are known to be efficiently computable. Towards understanding the limits of tractability in polymatrix games, we study the computation of Nash equilibria in such games where each pair of players plays either a zero-sum or a coordination game. We are particularly interested in the setting where players can be grouped into a small number of teams of identical interest. While the three-team version of the problem is known to be $\PPAD$-complete, the complexity for two teams has remained open. Our main contribution is to prove that the two-team version remains hard, namely it is $\CLS$-hard. Furthermore, we show that this lower bound is tight (i.e., $\CLS$-membership) for the setting where one of the teams consists of multiple independent adversaries. By leveraging this result we also obtain a simple algorithm that finds an $\varepsilon$-Nash equilibrium and only has a $1/\varepsilon^2$ dependence in $\varepsilon$ in its running time. On the way to obtaining our main result, we prove hardness of finding any stationary point in the simplest type of non-convex-concave min-max constrained optimization problem, namely for a class of bilinear polynomial objective functions.
\end{quote}

\section{Introduction}

Game theory is a fundamental tool to encode strategic agent interactions and has found many applications in the modern AI landscape such as Generative Adversarial Networks \cite{goodfellow2020generative}, obtaining agents with expert level play in multiplayer games such as Starcraft and Quake III \cite{vinyals2019grandmaster,jaderberg2019human} and superhuman performance in poker \cite{brown2018superhuman, brown2019superhuman}. Computing a Nash equilibrium or a saddle point (when considering general minmax optimization problems) is a computational task of central importance in these applications. The celebrated minmax theorem of von Neumann and Morgenstern~\cite{von2007theory} established that two-player zero-sum games have efficient algorithms. However, it was shown that three-player zero-sum games \cite{daskalakis2009complexity} or two-player general games \cite{chen2009settling} are computationally intractable (formally $\PPAD$-hard) and the hardness is also known to hold for computing approximations.

Consequently, Daskalakis and Papadimitriou~\cite{daskalakis2009network} proposed a tractable class of multiplayer zero-sum games, where the players are placed on the nodes of a graph and play a matrix zero-sum game with each adjacent player. In this setting, the total utility that a player gets is the sum of the utilities from each game that they participate in. It is to be highlighted that removing the zero-sum game assumption between each player makes the problem hard. Indeed, computing Nash equilibria in general polymatrix games is known to be $\PPAD$-hard \cite{daskalakis2009complexity,chen2009settling,Rubinstein18-Nash-inapproximability,DeligkasFHM24-pure-circuit}. Cai and Daskalakis~\cite{Cai11} studied an intermediate setting where every edge of the polymatrix game can be either a zero-sum game or a coordination game. They showed $\PPAD$-hardness, even for the special case where the players can be grouped into \emph{three} teams, such that players within the same team play coordination games, and players in different teams play zero-sum games. However the case of two-team polymatrix games with zero-sum and coordination edges has remained open.

\paragraph{Adversarial Two-team Games With a Single Adversary:}
The general notion of adversarial \emph{team} games introduced by Von Stengel and Koller~\cite{von1997team} studies two teams that are playing a zero-sum game with each other, meaning each team member gets the \emph{same} payoff as the whole team and the sum of the team payoffs is zero. The primary motivation here is to study strategies for companies against adversaries. What makes the ``team'' aspect special is that the team members cannot coordinate their actions and must play independent mixed strategies. This captures imperfect coordination within companies. Indeed, if the teams instead had perfect coordination, then the setting would simply reduce to a two player zero-sum game. Von Stengel and Koller showed that there exists a team maxmin equilibrium, that can be extended to a Nash equilibrium for this setting when the team is playing against a \emph{single} adversary. Moreover, they showed that this is the \emph{best} Nash equilibrium for the team, thereby alleviating equilibrium selection issues. However, it was later shown that finding a team maxmin equilibrium is in fact $\FNP$-hard and the problem does not become easier if one allows approximate solutions \cite{hansen2008approximability,borgs2008myth}. Recently, Anagnostides et al.~\cite{Anagnostides23} studied the \emph{single} adversary setting, and were able to show that finding a Nash equilibrium in this setting is in fact $\CLS$-complete. The $\CLS$-hardness immediately follows from the work of Babichenko and Rubinstein~\cite{BabichenkoR21-congestion}, but importantly it requires a sufficiently general game structure and so does not apply to the polymatrix setting. The $\CLS$-membership on the other hand applies to any adversarial game with a single adversary. The main idea is to obtain a Nash equilibrium from an approximate stationary point of the max Moreau envelope of the function $x \mapsto \max_{y \in \mathcal{Y}} U(x,y)$ (where $x$ is the min variable, $y$ is the max variable and $U$ is the payoff function).

\paragraph{Connections to Complexity of Minmax Optimization:}
Two-team games are a special case of general nonconvex-nonconcave constrained minmax problems. Daskalakis et al.~\cite{daskalakis2021complexity} recently studied this general setting and showed that finding a stationary point is $\PPAD$-complete. Crucially, their $\PPAD$-hardness only applies when the constraint sets are \emph{coupled} between the min and the max player. However, games usually induce minmax problems with \emph{uncoupled} constraints. The complexity of the problem for uncoupled constraints remains open,\footnote{In recent work that appeared after the first publication of our work, it was in fact shown that finding a KKT point of a general nonconvex-nonconcave constrained minmax problem is $\PPAD$-complete, even for uncoupled constraints~\cite{BernasconiC26-uncoupled-minmax}.} although it is known to be $\CLS$-hard, since it is at least as hard as finding stationary points of standard non-convex minimization problems \cite{fearnley2022complexity}. As we discuss below, our results also have implications for uncoupled minmax optimization, where we obtain a $\CLS$-hardness result for a particularly simple family of objective functions. We note that Li et al.~\cite{li2021complexity} showed a \emph{query} lower bound of $\Omega(\frac{1}{\varepsilon^2})$ for smooth nonconvex-strongly-concave minmax optimization problems, but these results do not apply to the simple objective functions that we study (and which can only be studied from the perspective of computational complexity).

\paragraph{Connections to Multiagent Learning:} From a learning dynamics perspective, qualitative results focus on understanding the limit behavior of certain no-regret learning dynamics in polymatrix games. In particular, some works focus on obtaining asymptotic convergence guarantees for Q-learning and its variants \cite{leonardos2021exploration, hussain2023asymptotic}. In a similar vein, some other works studied the limit behaviors of replicator dynamics for polymatrix games, particularly with zero-sum and coordination edges \cite{nagarajan2018three, nagarajan2020chaos}. In these works, the focus was on trying to identify network topologies under which the learning dynamics exhibited simple (non-chaotic) behaviors. Surprisingly, there were works that could obtain non-asymptotic convergence guarantees using discrete time algorithms in multiagent reinforcement learning as well, with Leonardos et al.~\cite{leonardos2021global} establishing convergence to Nash policies in Markov potential games. In adversarial settings, Daskalakis et al.~\cite{daskalakis2020independent} studied independent policy gradient and proved convergence to Nash policies. Moreover, some recent works establish convergence to Nash policies in Markov zero-sum team games~\cite{kalogiannis2022efficiently} and Markov polymatrix zero-sum games~\cite{kalogiannis2024zero}. This further establishes the need to theoretically study the computational challenges in the simplest polymatrix settings which allow for \emph{both} zero-sum and coordination edges, in order to understand convergence guarantees in more complicated multiagent reinforcement learning scenarios. This leads us to the following main question:
\begin{center}
    \emph{What is the complexity of finding Nash equilibria in two-team zero-sum polymatrix games?}
\end{center}

\subsection{Our Contributions}

Our main contribution is the following computational hardness result.

\begin{theorem}[Informal]\label{theorem:CLS-hardness}
It is $\CLS$-hard to find an approximate Nash equilibrium of a two-team zero-sum polymatrix game, even when one of the teams does not have any internal edges.
\end{theorem}

Our result is incomparable to the $\CLS$-hardness result proved by Anagnostides et al.~\cite{Anagnostides23} (which essentially immediately follows from the work of Babichenko and Rubinstein~\cite{BabichenkoR21-congestion}). On the one hand, our result is stronger because it applies to games with a simpler structure, namely polymatrix games, whereas their result only applies to the more general class of degree-5 polytensor games. On the other hand, our result is weaker because it requires the presence of multiple adversaries, instead of just a single adversary. Thus, our paper left open the case of two-team polymatrix games with a single adversary, which was resolved in subsequent work by Anagnostides et al.~\cite{anagnostides2025symmetric}.

As our second contribution, we complement the hardness result in \cref{theorem:CLS-hardness} by showing that the problem is in fact $\CLS$-complete in this particular case where the adversaries are independent (i.e., when there are no internal edges in the second team). Namely, the problem of finding an approximate Nash equilibrium in a two-team zero-sum polymatrix game with multiple independent adversaries lies in the class $\CLS$. The polymatrix setting allows us to provide a simple proof of this fact, in particular avoiding the use of more advanced machinery, such as the Moreau envelope used in the \CLS-membership of Anagnostides et al.~\cite{Anagnostides23}. Furthermore, leveraging the above result we show that a simple algorithm that utilizes (projected) gradient descent, finds an $\varepsilon$-Nash equilibrium and only has a $1/\varepsilon^2$ dependence in $\varepsilon$ in its running time. This is a quadratic improvement in the $\varepsilon$-dependence over the algorithm of Anagnostides et al.~\cite{Anagnostides23}, though their setting is not restricted to polymatrix games.

Going back to our main result, \cref{theorem:CLS-hardness}, we note that it also has some interesting consequences for minmax optimization in general. Namely, we obtain that computing a stationary point, i.e., a Karush-Kuhn-Tucker (KKT) point of
\[
\min_{x \in [0,1]^n} \max_{y \in [0,1]^m} f(x,y)
\]
is $\CLS$-hard, and thus intractable in the worst case, even when $f$ is a bilinear polynomial\footnote{Previous $\CLS$-hardness results required more general objective functions, namely $f$ had to be a quadratic (non-bilinear) polynomial \cite{Fearnley23}.} that is concave in $y$. This is somewhat surprising, as these objective functions are the simplest case beyond the well-known tractable setting of convex-concave.

\paragraph{The meaning of $\CLS$-hardness.}
The complexity class $\CLS$ was introduced by Daskalakis and Papadimitriou~\cite{DaskalakisP11-CLS} to capture the complexity of problems that are guaranteed to have a solution both by a fixed point argument and a local search argument. This class is a subset of two well-known classes: $\PPAD$ and $\PLS$. While $\PPAD$ is mainly known for capturing the complexity of computing Nash equilibria in general games \cite{daskalakis2009complexity,chen2009settling}, $\PLS$ captures the complexity of various hard local search problems, such as finding a locally maximal cut in a graph \cite{SchaefferY91-local-search}. Recently, following the result by Fearnley et al.~\cite{fearnley2022complexity} that $\CLS = \PPAD \cap \PLS$, it has been shown that the class captures the complexity of computing mixed Nash equilibria in congestion games \cite{BabichenkoR21-congestion}, and KKT points of quadratic polynomials \cite{Fearnley23}. See \cref{fig:tfnp-classes} for an illustration of the relationship between the classes.

A $\CLS$-hardness result indicates that the problem is very unlikely to admit a polynomial-time algorithm. To be more precise, our results indicate that we should not expect an algorithm to exist which can find an $\eps$-approximate Nash equilibrium in these two-team polymatrix games in time polynomial in $\log(1/\eps)$. Similarly, for the minmax problem with bilinear polynomial objective functions, we should not expect an algorithm that finds an $\eps$-KKT point in time polynomial in $\log(1/\eps)$. The evidence that $\CLS$ is a hard class is supported by various cryptographic lower bounds, which apply to both $\PPAD$ and $\CLS$ \cite{BitanskyPR15-Nash-crypto,ChoudhuriHKPRR19-Fiat-Shamir,JawaleKKZ21-PPAD-LWE}.

\begin{figure}
\centering
\begin{tikzpicture}
\node at (0,-2.3) {\textbf{\TFNP}};
\node at (-1.7,-3.75) {\textbf{\PLS}};
\node at (1.7,-3.75) {\textbf{\PPAD}};
\node at (0,-5) {\textbf{\CLS}};
\node at (0,-6.25) {\textbf{\P}};

\draw[-Stealth, very thick, gray] (1.3,-3.45) -- (0.5,-2.7);
\draw[-Stealth, very thick, gray] (-1.3,-3.45) -- (-0.5,-2.7);
\draw[-Stealth, very thick, gray] (-0.4,-4.7) -- (-1.2,-4.05);
\draw[-Stealth, very thick, gray] (0.4,-4.7) -- (1.2,-4.05);
\draw[-Stealth, very thick, gray] (0,-5.95) -- (0,-5.3);
\end{tikzpicture}
\caption{Classes of total search problems. Arrows are used to denote containment. For example, $\CLS$ is contained in $\PLS$ and in $\PPAD$. The class $\TFNP$ contains all total search problems, i.e., problems which are guaranteed to have efficiently checkable solutions. $\P$ contains all such problems solvable in polynomial time.}
\label{fig:tfnp-classes}
\end{figure}
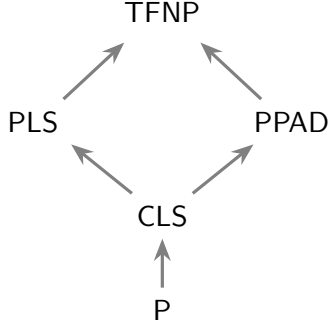

\subsection{Other Related Work}


\paragraph{Why two-team adversarial games/polymatrix games are interesting?} The study of team games was initiated by \cite{von1997team}, to model ``imperfect'' coordination within a company, when having to take strategic decisions in the presence of adversaries. In the field of AI agents, one can imagine such interactions are natural in settings where AI agents are trained to play team games, such as Starcraft \cite{vinyals2019grandmaster} and DoTA \cite{Berner2019}.

Meanwhile, polymatrix games are used to model pairwise interactions between players and these interactions can be specified as a graph. In some cases, polymatrix games offer tractable alternative models for multiplayer games, as NE in polymatrix zero-sum games are efficiently computable \cite{daskalakis2009network,Cai2016}. More generally, polymatrix games are used to model problems such as coordination games on graphs \cite{Apt2017, Apt2022} and this has applications in semi-supervised learning methods such as graph transduction \cite{Erdem2011}.

\paragraph{Additional related work on team games:}
It is worth mentioning that the computation of team maxmin equilibria for the two-team adversarial game has been studied where the team members are allowed to coordinate \emph{ex ante}, which is different from what a polymatrix coordination game would induce, as we study in this paper. There it was shown that the game can be reduced to a two-player game with imperfect-recall \cite{farina2018ex, celli2018computational, zhang2022subgame} and that efficient algorithms exist under some assumptions about the players' information sets \cite{zhang2021computing}. Finally, similar notions have been studied for extensive-form games too \cite{zhang2022team}. 


\paragraph{Linear Convergence for Bilinear/Convex-Concave Minmax}
Although we study this problem motivated by two-team polymatrix games with adversaries, the hardness results that we show also apply to the simplest non-convex concave minmax problem, i.e., bilinear non-convex concave. Our results strictly rule out the possibility of obtaining linear convergence. In contrast, bilinear zero-sum games and convex-concave games admit algorithms that converge linearly to the NE, for example see \cite{wei2020linear,lei2021last, sokota2022unified, liu2022power}.

\section{Preliminaries}

\subsection{Polymatrix games}

A polymatrix game~\cite{Janovskaja1968-polymatrix} is a type of multiplayer game in which the payoff function can be succinctly represented. More precisely, there is a set of players $\mathcal{N}$ and a set of pure strategies $S_i$ for each player $i \in \mathcal{N}$. Moreover, players are represented by the vertices of an undirected graph $G = (\mathcal{N},E)$ with the intent that a matrix game is played between each pair of players $\{i,\, j\} \in E$; the pair of payoff matrices being denoted by $A^{i,j} \in \mathbb{R}^{S_i \times S_j}$ and $A^{j,i} \in \mathbb{R}^{S_j \times S_i}$. Players are allowed to randomize over their pure strategies and play a mixed strategy in the probability simplex of their action, which we denote by $\Delta(S_i)$ for player $i$. Hence, a mixed strategy profile is some $x = (x_1,\, x_2,\, \ldots,\, x_{|\mathcal{N}|}) \in \mathcal{X} \coloneqq \times_{i \in \mathcal{N}} \Delta(S_i)$. We also use the standard notation $x_{-i}$ to represent the mixed strategies of all players other than $i$. In polymatrix games, the utility $U_i:\mathcal{X} \rightarrow \mathbb{R}$ for player $i\in \mathcal{N}$ is the sum of her payoffs, so that
\[
U_i(x)=\sum_{j: \{i,j\} \in E}x_iA^{i,j}x_j
\]
where $x_i$ is understood to be $x_i^\top$: we drop the transpose when it is clear from the context for ease of presentation. As a well-defined class of games, polymatrix games always admit a Nash equilibrium \cite{nash1951non}. In this work, we are interested in \emph{approximate} Nash equilibria which we define next.

\begin{definition}\label{definition:approx-NE}
Let $\epsilon \geq 0$ be an approximation guarantee. The mixed strategy profile $\widetilde{x}$ is an $\epsilon$-approximate Nash equilibrium of the polymatrix game defined above if for any $i \in \mathcal{N}$, 
\[
 U_i(x_i,\widetilde{x}_{-i}) \leq U_i(\widetilde{x}_i,\widetilde{x}_{-i}) + \epsilon \quad \forall x_i \in \Delta(S_i)
\]
\end{definition}

In this paper, we focus on polymatrix games with a particular structure where players are grouped into two competing teams with players within teams sharing mutual interests.

\begin{definition}[Two-team Polymatrix Zero-Sum Game]\label{def:two-team}
A two-team polymatrix zero-sum game is a polymatrix game $\{A^{i, j}\}_{i, j \in \mathcal{N}}$ where the players can be split into two teams $X \cup Y = \mathcal{N}$ so that any game between the two teams is zero-sum and any game within a team is a coordination game. More precisely for any $i, i' \in X$ and $j, j' \in Y$:
\[
A^{i, i'} = \big( A^{i', i}\big)^\top \quad A^{j, j'} = \big( A^{j', j}\big)^\top \quad A^{i, j} = -\big(A^{j, i}\big)^\top
\]
If there is no coordination within team $Y$, that is $A^{j, j'} = 0^{S_j \times S_{j'}} $ for every $j, j' \in Y$, we further say that it is a two-team zero-sum game with independent adversaries.
\end{definition}

In the restricted context of two-teams games, another useful equilibrium concept is that of \emph{team-maxmin equilibria} \cite{von1997team}. While originally defined for a single adversary only, it is generalizable to multiple independent adversaries in such a way that any team-maxmin equilibrium can be converted to a Nash equilibrium efficiently. Unfortunately, similarly to the single-adversary case, such equilibria suffer from intractability issues and are $\FNP$-hard to compute \cite{basilico2017team, hansen2008approximability, borgs2008myth}.

\subsection{KKT Points of Constrained Optimization Problems}

We now turn our attention to solution concepts for optimization problems and in particular of Karush-Kuhn-Tucker (KKT) points. We only define the required notions for the special case where each variable is constrained to be in $[0, 1]$, since this will be sufficient for us. Under those \emph{box constraints} the expression of the Lagrangian simplifies greatly and we obtain the following definition of approximate KKT point for a minimization problem.

\begin{definition}\label{def:min-kkt}
Let $\epsilon \geq 0$ be an approximation parameter, $f: \mathbb{R}^n \to \mathbb{R}$ a continuously differentiable function and consider the optimization problem $\min\nolimits_{x \in [0, 1]^n} f(x)$. The point $\widetilde{x} \in [0, 1]^{n}$ is an $\epsilon$-approximate \textup{KKT} point of the formulation if the gradient $g \coloneqq \nabla f(\widetilde{x})$ satisfies for each $i \in [n]$:
\begin{enumerate}
    \item If $\widetilde{x}_i \in (0, 1)$, then $|g_i| \leq \epsilon$.
    \item If $\widetilde{x}_i = 0$, then $g_i \geq -\epsilon$.
    \item If $\widetilde{x}_i = 1$, then $g_i \leq \epsilon$.
\end{enumerate}
\end{definition}

Thus, an exact ($\epsilon = 0$) KKT point can be thought of as a fixed point of the gradient descent algorithm.  Using this intuition, we can extend this to minmax problems as fixed points of the gradient \emph{descent-ascent} algorithm. See \cref{fig:minmaxkkt} for the geometric intuition of minmax KKT points.

\begin{definition}\label{def:minmax-kkt}
Let $\epsilon \geq 0$ be an approximation parameter, $f: \mathbb{R}^n\times \mathbb{R}^n \to \mathbb{R}$ a continuously differentiable function and consider the optimization problem $\min\nolimits_{x \in [0, 1]^n}\max\nolimits_{y \in [0, 1]^n} f(x, y)$. Let $(\widetilde{x}, \widetilde{y}) \in [0, 1]^{2n}$ and let $(g, q) \coloneqq \nabla f(\widetilde{x}, \widetilde{y})$, where $g$ is the gradient with respect to $x$-variables and $q$ with respect to $y$-variables. Then, $(\widetilde{x}, \widetilde{y})$ is an \emph{$\epsilon$-approximate \textup{KKT} point} if for each $i \in [n]$:
\begin{multicols}{2}
    \begin{enumerate}
        \item If $x_i \in (0, 1)$, then $|g_i| \leq \epsilon$.
        \item If $x_i = 0$, then $g_i \geq -\epsilon$.
        \item If $x_i = 1$, then $g_i \leq \epsilon$.
        \item If $y_i \in (0, 1)$, then $|q_i| \leq \epsilon$.
        \item If $y_i = 0$, then $q_i \leq \epsilon$.
        \item If $y_i = 1$, then $q_i \geq -\epsilon$.
    \end{enumerate}
\end{multicols}
\end{definition}

\begin{figure}
    \centering
    \begin{tikzpicture}

\path[fill = blue!10!white, draw=gray] (0, 0) -- (3, 0) -- (3, 3) -- (0, 3) -- cycle;

\path[->, thick] (-.5, 0) edge (3.5, 0);
\path[->, thick] (0, -.5) edge (0, 3.5);
\node[] at (0, 3.7) {\small $y$};
\node[] at (3.7, 0) {\small $x$};
\node[] at (3, -.3) {\small $1$};
\node[] at (-.3, 3) {\small $1$};
\node[] at (-.45, -.3) {\small $(0, 0)$};

\path[-Stealth, draw=black!30!green, thick] (0, 2) node[fill=black!30!green, circle, minimum height=4, inner sep = 0]{} -- (.5, 2);
\path[-Stealth, draw=black!30!green, thick] (1, 3) node[fill=black!30!green, circle, minimum height=4, inner sep = 0]{} -- (1, 3.7);
\path[-Stealth, draw=black!30!green, thick] (1, 3) node[fill=black!30!green, circle, minimum height=4, inner sep = 0]{} -- (1, 3.7);
\path[-Stealth, draw=black!30!green, thick] (3, 3) node[fill=black!30!green, circle, minimum height=4, inner sep = 0]{} -- (2.4, 3.3);
\path[-Stealth, draw=black!30!green, thick] (3, 2) node[fill=black!30!green, circle, minimum height=4, inner sep = 0]{} -- (2.3, 2);
\node[] at (3.66, 2.1) {\small $(x_2, y_2)$}; 
\node[fill=black!30!green, circle, minimum height=4, inner sep = 0] at (2.6, 1.2) {};


\path[-Stealth, draw=black!30!red, thick] (0, .5) node[fill=black!30!red, circle, minimum height=4, inner sep = 0]{} -- (.75, 1);
\node[] at (-.66, .55) {\small $(x_1, y_1)$}; 
\path[-Stealth, draw=black!30!red, thick] (1.7, 3) node[fill=black!30!red, circle, minimum height=4, inner sep = 0]{} -- (1.3, 2.7);
\path[-Stealth, draw=black!30!red, thick] (1.7, 0) node[fill=black!30!red, circle, minimum height=4, inner sep = 0]{} -- (1.7, 0.9);
\path[-Stealth, draw=black!30!red, thick] (1.5, 1.5) node[fill=black!30!red, circle, minimum height=4, inner sep = 0]{} -- (1.8, 2.5); 
\end{tikzpicture}
    \caption{The intuition behind exact ($\epsilon = 0$) KKT points for minmax problems. The formulation features a min-variable $x$, a max-variable $y$ and the bounding box constraint $(x, y) \in [0, 1]^2$. Some feasible points together with their respective gradient $(q^x, q^y)$ are depicted. Green points are valid KKT points whereas red ones are not. For instance, $(x_1, y_1)$ is not a KKT point because $q_1^y > 0$ but $y < 1$. On the other hand, $(x_2, y_2)$ is a valid KKT point because $q^y_2 = 0$ and $q^x_2 < 0$ with $x = 1$.}
    \label{fig:minmaxkkt}
\end{figure}
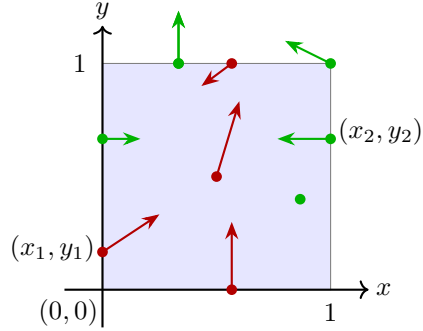

\subsection{Connection between two-team games and minmax optimization}

Given a two-team polymatrix zero-sum game with matrices $\{A^{i, j}\}_{i, j \in \mathcal{N}}$, we can define the following common utility function
$$U(x,y) = - \sum_{i, i' \in X: i < i'} x_i A^{i,i'} x_{i'} - \sum_{i \in X, j \in Y} x_i A^{i,j} y_j$$
and the corresponding game where players on the $X$-team all have the same utility function $-U$, and players on the $Y$-team all have the same utility function $U$.
It is easy to check that this new game is equivalent to the previous one, namely a strategy profile $(x,y)$ is an $\eps$-Nash equilibrium in the latter game if and only if it is in the former. Now, if we consider the minmax optimization problem
\[
\min_{x \in \mathcal{X}}\max_{y \in \mathcal{Y}} U(x,y)
\]
it is not hard to show that its KKT points exactly correspond to the set of Nash equilibria of the game with common utility, and thus also of the original polymatrix game. This connection also extends to approximate solutions. We develop this relation in greater detail in section \cref{sec:hardness}, where we first prove a hardness result for a minmax problem and then use the equivalence above to extend it to polymatrix games.

\section{\texorpdfstring{$\CLS$}{CLS}-hardness of Two-Team Games with Independent Adversaries}
\label{sec:hardness}

In this section, we give a proof of our main hardness result which we re-state formally next.

\begin{theorem}[Precise formulation of \cref{theorem:CLS-hardness}]
It is $\CLS$-hard to find an $\epsilon$-approximate Nash equilibrium of a two-team zero-sum polymatrix game with independent adversaries where $\epsilon$ is inverse exponential in the description of the game.
\end{theorem}

We show that this is true even when each player can only choose from two pure strategies. This lower bound result is obtained by reducing from the problem of finding some KKT point of a quadratic minimization problem. In detail, an instance of $\MinQuadKKT$ consists of a quadratic polynomial $Q: \mathbb{R}^n \to \mathbb{R}$ together with an approximation parameter $\epsilon > 0$. A solution to such instance is any $\epsilon$-approximate KKT point of $\min_{x \in [0, 1]^n} Q(x)$ (see \cref{def:min-kkt}). $\MinQuadKKT$ is known to be $\CLS$-complete for $\epsilon$ inverse exponential in the description of the instance~\cite{Fearnley23}. The reduction from $\MinQuadKKT$ to two-team games is performed in two stages which we describe next.

\paragraph{\textbf{Stage I.}} As a first step, we show how to reduce $\MinQuadKKT$ to the intermediate problem $\MinmaxIndKKT$. An instance of $\MinmaxIndKKT$ consists of an approximation parameter $\epsilon > 0$ together with a polynomial $M: \mathbb{R}^{2n} \to \mathbb{R}$ satisfying the following three properties:
\begin{enumerate}
\item $M$ is multilinear.
\item $M$ has degree at most 2.
\item $M(x, y)$ has no monomial of type $y_iy_j$ for $i, j \in [n]$.
\end{enumerate}
In this section, we conveniently refer to those three conditions as the \emph{independence} property. A solution to such an instance is any $\epsilon$-approximate KKT point of $\min_{x \in [0,1]^n} \max_{y \in [0,1]^n} M(x, y)$ (see \cref{def:minmax-kkt}). Let us highlight that this step already establishes the $\CLS$-hardness of uncoupled minmax multilinear formulations for a simple class of objectives.

\paragraph{\textbf{Stage II.}} We exploit the independence property of the polynomial generated by the first stage reduction to reduce it further to a two-team zero-sum polymatrix game with independent adversaries. This is achieved through generalizing the equivalence between zero-sum games for two players and some class of minmax formulations.

\subsection{Stage I: from quadratic to multilinear}

\begin{lemma}\label{lemma:quadratic_to_multilinear}
$\MinQuadKKT$ reduces to $\MinmaxIndKKT$
\end{lemma} 
By ``reduces to'', we mean the usual definition of polynomial-time $\TFNP$ reduction. We begin by describing the reduction in detail, then highlight that it yields an objective with the independence property and finally prove correctness. Let $Q: [0, 1]^n \to \mathbb{R}$ and $\epsilon > 0$ be the $\MinQuadKKT$ instance. Write the coefficients of $Q$ explicitly as follows:
\[
Q(x) = q + \sum\nolimits_{i \in [n]} q_ix_i + \sum\nolimits_{i \neq j} q_{ij} x_ix_j + \sum\nolimits_{i \in [n]} q_{ii} x_i^2
\]
Let $Z \geq 1$ be an upper bound on the sum of the absolute values of all coefficients of $Q$ and note that since $Q$ is a quadratic polynomial, it holds that for every $x \in [0, 1]^n$ and $i \in [n]$:
\begin{equation}\label{eq:bound_on_quadratic}
|Q(x)| \leq Z \quad\text{and}\quad \left\lvert \frac{\partial Q(x)}{\partial x_i} \right\rvert \leq 2Z
\end{equation}
Fix $T \coloneqq 10Z$ and  $\eta \coloneqq \epsilon/T$. For each min-variable $x_i$ of the $\MinQuadKKT$ instance, we introduce two min-variables $x_i$ and $x_i'$ and a max-variable $y_i$ in the $\MinmaxIndKKT$ instance. We also define the following multilinear ``copy'' gadget:
\[
\copyGadget(x_i, x_i', y_i) \coloneqq \Big(x_i' - x_i \cdot (1 - 2\eta) - \eta\Big) \cdot (y_i - 1/2)
\]
The role of $\copyGadget$ is to force $x_i$ to be close to $x_i'$ at any KKT point, thus effectively \emph{duplicating} $x_i$ into $x_i'$. This allows us to remove square terms of the objective function. To make this formal, let $Q': [0,1]^{2n} \to \mathbb{R}$ be a copy of $Q$ where every occurrence of the form $x_i^2$ is replaced by $x_ix_i'$:
\[
Q'(x, x') \coloneqq q + \sum\nolimits_{i \in [n]} q_ix_i + \sum\nolimits_{i \neq j} q_{ij} x_ix_j + \sum\nolimits_{i \in [n]} q_{ii} x_ix_i'
\]
The full formulation of the $\MinmaxIndKKT$ instance is stated below. We note that the objective $M$ of the formulation indeed satisfies the independence property.\footnote{Additional dummy max-variables can be added to ensure that the number of min- and max-variables are the same.} The proof of \cref{lemma:quadratic_to_multilinear} now follows from \cref{claim:CLS-hardness-main}.
\[
\min_{x, x' \in [0,1]^n} \max_{y \in [0,1]^n} M(x, x', y) \quad\text{where}\quad M(x, x', y) \coloneqq Q'(x, x') + \sum_{i \in [n]} T \cdot \textsc{copy}(x_i, x_i', y_i)
\]

\begin{claim}
\label{claim:CLS-hardness-main}
For any $\epsilon \in (0, 1]$, if $(x, x', y)$ is an $(\epsilon/10)$-approximate KKT point of the $\MinmaxIndKKT$ instance, then $x$ is an $\epsilon$-approximate KKT point of the $\MinQuadKKT$ instance.
\end{claim}

The proof relies crucially on $T$ being sufficiently large so that the copy gadgets dominate $Q'(x, x')$ and the objective $M$. This forces any KKT point $(x, x', y)$ to have $y \in (0,1)^n$ and ultimately $x \approx x'$. A second step shows that $\partial Q(x) / \partial x_i \approx \partial M(x, x', y) / \partial x_i$, which is enough to conclude that $x$ is a KKT point of the $\MinQuadKKT$ instance.

\begin{proof}
Let $\delta \coloneqq \epsilon/10$ be the $\MinmaxIndKKT$ approximation guarantee for the remainder of this argument. Recall that $\eta = \epsilon/T = \delta/Z$. As a first step, we show that $y_i \in (0, 1)$ for each $i \in [n]$. Towards a contradiction, let us first suppose that $y_i = 1$ and observe that:
\begin{align*}
    \frac{\partial M(x, x', y)}{\partial x_i'} &= \frac{\partial Q'(x, x')}{\partial x_i'} + T \cdot (y_i - 1/2)\\
    &= \frac{\partial Q'(x, x')}{\partial x_i'} + T/2\\
    &\geq -2Z + T/2 &\text{using \eqref{eq:bound_on_quadratic}}\\
    &> \delta &\text{as $T = 10Z$ and $Z \geq 1$}
\end{align*}
But as $(x, x', y)$ is a $\delta$-approximate KKT point and $x_i'$ is a min-variable, it must be that $x_i' = 0$ (see \cref{def:minmax-kkt}). As a result, we obtain:
\[
\frac{\partial M(x, x', y)}{\partial y_i} = T (x_i' - x_i \cdot (1 - 2\eta) - \eta) \leq - T \eta < - \delta
\]
where we used the fact that $x_i \geq 0$ and $\eta = \delta/Z > \delta/T$.
Now, because $y_i$ is a max-variable and $(x, x', y)$ a $\delta$-KKT point, it must be that $y_i = 0$: a contradiction to the assumption that $y_i = 1$. One can rule out the possibility of $y_i = 0$ in a similar way and we may thus conclude that $y_i \in (0, 1)$. This fact, together with the KKT conditions, further implies that the partial derivative of $M$ with respect to $y_i$ almost vanishes for each $i \in [n]$, i.e.,
\begin{equation}\label{eq:CLS-hardness-smallalpha}
T \left\lvert x_i' - x_i \cdot (1 - 2\eta) - \eta \right\rvert = \left\lvert\frac{\partial M(x, x', y)}{\partial y_i} \right\rvert \leq \delta
\end{equation}
This shows that $x_i' \in [\eta - \delta/T, 1 - \eta + \delta/T]$ and combining with $\eta > \delta/T$ it follows that $x_i' \in (0, 1)$ and the KKT conditions imply:
\[
\delta \geq \left\lvert \frac{\partial M(x, x', y)}{\partial x_i'} \right\rvert = \left \lvert \frac{\partial Q'(x, x')}{\partial x_i'} + T \cdot (y_i - 1/2) \right\rvert \,\implies\, T \cdot (1/2 - y_i) = \frac{\partial Q'(x, x')}{\partial x_i'} \pm \delta
\]
Where we use the notation $a = b \pm \delta$ to mean $a \in [b - \delta, b + \delta]$. We can now show that the partial derivative of $M$ at $(x, x', y)$ is very close to the one of $Q$ at $x$ for all coordinates $x_i$:
\begin{align*}
    \frac{\partial M(x, x', y)}{\partial x_i} &= \frac{\partial Q'(x, x')}{\partial x_i} + (1 - 2\eta) \cdot T \cdot (1/2 - y_i)\\
    &= \frac{\partial Q'(x, x')}{\partial x_i} + (1 - 2\eta) \cdot \left( \frac{\partial Q'(x, x')}{\partial x_i'} \pm \delta \right)\\
    &= \frac{\partial Q'(x, x')}{\partial x_i} +  \frac{\partial Q'(x, x')}{\partial x_i'}   \pm \left( 2\eta \left\lvert \frac{\partial Q'(x, x')}{\partial x_i'} \right\rvert +  \delta \right)\\
    &= \frac{\partial Q'(x, x')}{\partial x_i} +  \frac{\partial Q'(x, x')}{\partial x_i'}  \pm \left( 4\eta Z +  \delta \right)\\
    &= \frac{\partial Q'(x, x')}{\partial x_i} +  \frac{\partial Q'(x, x')}{\partial x_i'}  \pm 5 \delta &\text{as $\eta = \delta/Z$}
\end{align*}
Observe that \eqref{eq:CLS-hardness-smallalpha} also implies that $x_i$ and $x_i'$ must be close with $|x_i - x_i'| \leq \delta/T + \eta \leq 2\eta$, hence:
\begin{align*}
\frac{\partial Q'(x, x')}{\partial x_i} +  \frac{\partial Q'(x, x')}{\partial x_i'} &= q_i + \sum_{j \neq i} q_{ij} x_j + q_{ii} x_i + q_{ii} x_i'\\
&= \frac{\partial Q(x)}{\partial x_i} + q_{ii} (x_i' - x_i)\\
&= \frac{\partial Q(x)}{\partial x_i} \pm 2 \delta &\quad\text{as $|q_{ii}| \leq Z$ and $2\eta = 2\delta / Z$}
\end{align*}
Combining the two previous observations, we have $\partial Q(x)/\partial x_i = \partial M(x, x', y)/\partial x_i \pm 7 \delta$. With this fact established, we may finally show that $x$ is indeed an $\epsilon$-KKT point of the $\MinQuadKKT$ formulation. If $x_i < 1$, then note that:
\[
\frac{\partial Q(x)}{\partial x_i} \geq \frac{\partial M(x, x', y)}{\partial x_i} - 7 \delta \geq - \delta - 7 \delta \geq - 8 \delta \geq - \epsilon
\]
Similarly, if $x_i > 0$, we have:
\[
\frac{\partial Q(x)}{\partial x_i} \leq \frac{\partial M(x, x', y)}{\partial x_i} + 7 \delta \leq \delta + 7 \delta \leq 8 \delta \leq \epsilon. \qedhere
\]
\end{proof}

\subsection{Stage II: from minmax to two-team games}
To prove \cref{theorem:CLS-hardness}, we use stage I (\cref{lemma:quadratic_to_multilinear}) and show how to reduce an instance of $\MinmaxIndKKT$ to a two-team zero-sum game with independent adversaries. Let $\epsilon > 0$ be the approximation parameter of the instance and its objective be $M: \mathbb{R}^{2n} \to \mathbb{R}$. Since $M$ has the independence property, we can explicitly write its coefficient as:
\[
M(x, y) \coloneqq \alpha + \sum\nolimits_{i \in [n]} \beta_i x_i + \sum\nolimits_{i \neq j} \gamma_{ij} x_ix_j + \sum\nolimits_{i \in [n]} \zeta_i y_i + \sum\nolimits_{i, j \in [n]} \theta_{ij} x_iy_j
\]
We construct a polymatrix game with two teams, each consisting of $n$ players. In the first (cooperative) team, there is one player $a_i$ corresponding to each variable $x_i$. The second team consists of $n$ independent adversaries, where player $b_i$ corresponds to variable $y_i$. The intent is that an optimal strategy profile of the players roughly corresponds to a KKT point $(x, y)$ of $M$. As stated earlier, we reduce to a very restricted setting where each player only has two actions. We thus specify the utility matrices of the game as elements of $\mathbb{R}^{2 \times 2}$ with:
\begin{align}
A^{a_i, a_j} &= A^{a_j, a_i} = \begin{bmatrix} -\gamma_{ij} & 0 \\ 0 & 0 \end{bmatrix} &\quad \text{for all $i, j\in[n]$ with $i \neq j$} \label{eq:matrix_internal}\\
A^{b_j, a_i} &= - (A^{a_i,b_j})^\top = \begin{bmatrix} \theta_{ij} + \zeta_j/n + \beta_i/n & \zeta_j/n \\ \beta_i/n & 0 \end{bmatrix} &\quad \text{for all $i, j\in[n]$} \label{eq:matrix_across}
\end{align}
Any other utility matrix is set to $0^{2 \times 2}$. Observe that this payoff setting indeed yields a proper two-team zero-sum polymatrix game with independent adversaries. Let $Z \geq 1$ be an upper bound on the sum of absolute coefficients of $M$ and let $\delta \coloneqq \epsilon^2/(4Z)$ be the target approximation ratio for the polymatrix game. If $(p,q)$ is a $\delta$-approximate Nash equilibrium of the polymatrix game, we define a candidate KKT point $(x, y) \in \mathbb{R}^{2n}$ as follows:
\begin{align*}
    x_i = \begin{cases}
        0 &\quad\text{if $p_i < \epsilon/(2Z)$}\\
        1 &\quad\text{if $p_i > 1 - \epsilon/(2Z)$}\\
        p_i &\quad\text{else}
    \end{cases}
    \quad\quad\text{and}\quad\quad
    y_i = \begin{cases}
        0 &\quad\text{if $q_i < \epsilon/(2Z)$}\\
        1 &\quad\text{if $q_i > 1 - \epsilon/(2Z)$}\\
        q_i &\quad\text{else}
    \end{cases}
\end{align*}
Here $p_i \in [0,1]$ represents the probability that player $a_i$ plays its first action, and similarly $q_i$ for player $b_i$. The correctness of the reduction is treated in \cref{claim:correctness_stage_II} and thus \cref{theorem:CLS-hardness} follows.

\begin{claim}
\label{claim:correctness_stage_II}
$(x, y)$ is an $\epsilon$-approximate KKT point of the $\MinmaxIndKKT$ instance.
\end{claim}

\begin{proof}
We only show that $x$-variables satisfy the KKT conditions as the proof is similar for the $y$-variables. It follows from \eqref{eq:matrix_internal} and \eqref{eq:matrix_across} that for any $\widetilde{p}_i \in [0, 1]$ the utility for player $a_i$ can be written as:
\begin{equation*}
\begin{split}
U_{a_i}(\widetilde{p}_i, p_{-i}, q) &= \sum\nolimits_{j \neq i} \begin{bmatrix} \widetilde{p}_i \\ 1-\widetilde{p}_i \end{bmatrix}^\top A^{a_i,a_j} \begin{bmatrix} p_j \\ 1-p_j \end{bmatrix} + \sum\nolimits_{j \in [n]} \begin{bmatrix} \widetilde{p}_i \\ 1-\widetilde{p}_i \end{bmatrix}^\top A^{a_i,b_j} \begin{bmatrix} q_j \\ 1-q_j \end{bmatrix} \\
&= - \widetilde{p}_i \cdot \left(\beta_i + \sum\nolimits_{j \neq i} \gamma_{ij} p_j + \sum\nolimits_{j \in [n]} \theta_{ij} q_j \right) - \sum\nolimits_{j \in [n]} \zeta_j q_j /n
\end{split}
\end{equation*}
Since $(p,q)$ is a $\delta$-approximate Nash equilibrium, we can use the above expression twice and \cref{definition:approx-NE} to get that:
\begin{equation}\label{eq:eta_a}
(p_i - \widetilde{p}_i) \cdot \left(\beta_i + \sum\nolimits_{j \neq i} \gamma_{ij} p_j + \sum\nolimits_{j \in [n]} \theta_{ij} q_j  \right) \leq \delta \quad \forall\, \widetilde{p}_i \in [0,1]
\end{equation}
Fix some variable $x_i$ and let us verify that it satisfies the $\epsilon$-approximate KKT condition. The partial derivative of $M$ with respect to $x_i$ is:
\begin{align*}
\frac{\partial M(x, y)}{\partial x_i} &= \beta_i + \sum_{j \neq i} \gamma_{ij} x_j + \sum_{j \in [n]} \theta_{ij} y_j\\
&= \beta_i + \sum_{j \neq i} \gamma_{ij} \cdot (p_j + x_j - p_j) + \sum_{j \in [n]} \theta_{ij} \cdot (q_j + y_j - q_j) \\
&= \beta_i + \sum_{j \neq i} \gamma_{ij} p_j + \sum_{j \in [n]} \theta_{ij} q_j \pm \Bigg( \sum_{j \neq i} |\gamma_{ij}| \cdot |x_j - p_j| + \sum_{j \in [n]} |\theta_{ij}| \cdot |y_j - q_j| \Bigg)\\
&= \beta_i + \sum_{j \neq i} \gamma_{ij}p_j + \sum_{j \in [n]} \theta_{ij} q_j \pm \frac{\epsilon}{2Z} \cdot Z
\end{align*}
Where in the last equality, we used the fact that $x$ (respectively $y$) is close to $p$ (respectively $q$).

We now finish the proof by considering two cases for $x_i$. First, consider the case where $x_i < 1$. By definition of $x_i$, this implies that $p_i \leq 1 - \epsilon/2Z$. Thus, setting $\widetilde{p}_i \coloneqq 1$ in \eqref{eq:eta_a}, we get:
\begin{align*}
&(1 - p_i) \cdot \left(\beta_i + \sum\nolimits_{j \neq i} \gamma_{ij} p_j + \sum\nolimits_{j \in [n]} \theta_{ij} q_j  \right) \geq -\delta\\
&\quad\implies \beta_i + \sum\nolimits_{j \neq i} \gamma_{ij} p_j + \sum\nolimits_{j \in [n]} \theta_{ij} q_j \geq -\frac{\delta}{1 - p_i} \geq -\frac{2Z\cdot \delta}{\epsilon} \geq -\epsilon/2
\end{align*}
and thus $\partial M(x, y) / \partial x_i \geq -\epsilon/2 - \epsilon/2 \geq - \epsilon$. Next, consider the case where $x_i > 0$. By definition of $x_i$, this implies that $p_i \geq \epsilon/2Z$. Now, setting $\widetilde{p}_i \coloneqq 0$ in \eqref{eq:eta_a}, we get:
\begin{align*}
&p_i \cdot \left(\beta_i + \sum\nolimits_{j \neq i} \gamma_{ij} p_j + \sum\nolimits_{j \in [n]} \theta_{ij} q_j  \right) \leq \delta\\
&\quad\implies \beta_i + \sum\nolimits_{j \neq i} \gamma_{ij} p_j + \sum\nolimits_{j \in [n]} \theta_{ij} q_j \leq \frac{\delta}{p_i} \leq \frac{2Z\cdot \delta}{\epsilon} \leq \epsilon/2
\end{align*}
and thus $\partial M(x, y) / \partial x_i \leq \epsilon/2 + \epsilon/2 \leq \epsilon$. This shows that $x_i$ always satisfies the $\epsilon$-KKT conditions. In particular, when $x_i \in (0,1)$, $|\partial M(x, y) / \partial x_i| \leq \epsilon$.
\end{proof}

\section{\texorpdfstring{$\CLS$}{CLS}-membership and Approximation Algorithm for Independent Adversaries}

In this section we prove upper bounds for two-team zero-sum polymatrix games with independent adversaries. Namely, we first prove $\CLS$-membership for the problem and then use this formulation to obtain an approximation algorithm. Our approach leverages the additional structure of the polymatrix game to obtain a better running time than the algorithm in~\cite{Anagnostides23}, which applies to a class of two-team games that is more general than polymatrix.

\subsection{\texorpdfstring{$\CLS$}{CLS}-membership}
\label{sec:membership}

In this section we prove the following.

\begin{theorem}\label{thm:membership}
The problem of computing a Nash equilibrium in two-team zero-sum polymatrix games with independent adversaries lies in $\CLS$.
\end{theorem}

In particular, this implies that the $\CLS$-hardness result proved in the previous section is tight for such games.

\paragraph{\textbf{Reformulation as a minimization problem.}} The main idea to prove the theorem is to start from the minmax formulation of the problem and to rewrite it as a minimization problem by using duality. Let a two-team polymatrix zero-sum game with independent adversaries be given. Without loss of generality, we assume that every player has exactly $m$ strategies. Recall that by the structure of the game we have $A^{i,i'} = (A^{i',i})^\top$, $A^{i,j} = -(A^{j,i})^\top$, and $A^{j,j'} = 0$ for all $i,i' \in X$, $j,j' \in Y$ with $i \neq i'$ and $j \neq j'$. We can write

\begin{equation}\label{eq:minmax-to-min}
\begin{split}
&\min_{x \in \mathcal{X}}\max_{y \in \mathcal{Y}} - \sum_{i, i' \in X: i < i'} x_i A^{i,i'} x_{i'} - \sum_{i \in X, j \in Y} x_i A^{i,j} y_j\\
    = &\min_{x \in \mathcal{X}}\max_{y \in \mathcal{Y}} - \sum_{i, i' \in X: i < i'} x_i A^{i,i'} x_{i'} + \sum_{i \in X, j \in Y} y_j A^{j,i} x_i\\
    = &\min_{x \in \mathcal{X}}\left\{- \sum_{i, i' \in X: i < i'} x_i A^{i,i'} x_{i'} + \max_{y \in \mathcal{Y}}\sum_{i \in X, j \in Y} y_j A^{j,i} x_i\right\}
    \end{split}
\end{equation}
Now consider the ``max'' part of the above objective written as the following LP in $y$ variables:

\begin{align*}
\begin{split}
    \max\;\; &\sum_{j \in Y} c_j^\top y_j \\
     \;\;&\sum_{k=1}^m y_{jk} = 1 \;\; \forall j \in Y\\
    \;\;& y_{jk} \geq 0\;\;\forall j \in Y\;\;\text{and}\;\;\forall k \in [m]
    \end{split}
\end{align*}
where $c_j = \sum_{i \in X} A^{j,i} x_i$ for all $j \in Y$. Then the dual of the above program can be written as:
\begin{align*}
\begin{split}
    \min\;\; &\sum_{j \in Y}\gamma_j\\
     \;\;&\gamma_j \geq c_{jk} \;\; \forall j \in Y\;\;\text{and}\;\;\forall k \in [m]
    \end{split}
\end{align*}

Thus replacing the max part in \eqref{eq:minmax-to-min} by the equivalent dual formulation we obtain:

\begin{equation}\label{eqn:equivminformulation}
\begin{split}
\min \quad & - \sum_{i, i' \in X: i < i'} x_i A^{i,i'} x_{i'} + \sum_{j \in Y} \gamma_j \\
\text{s.t. } &\gamma_j \geq \sum_{i \in X} e_k^\top A^{j,i} x_i  \;\; \forall j \in Y\;\;\text{and}\;\;\forall k \in [m] \\
& \gamma_j \leq M \;\; \forall j \in Y \\
& \gamma \in \mathbb{R}^{|Y|} \\
& x \in \mathcal{X}
\end{split}
\end{equation}
where $e_k \in \mathbb{R}^m$ is the $k$th unit vector. We have introduced an additional set of constraints $\gamma_j \leq M$ to ensure that the feasible region is bounded. $M$ is chosen to be sufficiently large such that $M > \max_{x \in \mathcal{X}, j \in Y, k \in [m]} \sum_{i \in X} e_k^\top A^{j,i} x_i$. Note that in order to obtain the formulation \eqref{eqn:equivminformulation} we have crucially used the fact that the original game is polymatrix and has independent adversaries.

For what comes next, we will need the following definition of KKT points, which generalizes \cref{def:min-kkt} to arbitrary linear constraints.

\begin{definition}\label{def:min-KKT-general}
Consider an optimization problem of the form
\begin{equation*}
\begin{split}
\min \quad &f(x) \\
\text{s.t. } &Ax \leq b \\
&x \in \mathbb{R}^n
\end{split}
\end{equation*}
where $f: \mathbb{R}^n \to \mathbb{R}$ is continuously differentiable, $A \in \mathbb{R}^{m \times n}$, and $b \in \mathbb{R}^m$. A point $x^* \in \mathbb{R}^n$ is a KKT point of this problem if there exists $\mu \in \mathbb{R}^m$ such that
\begin{enumerate}
\item $\nabla f(x^*) + A^\top \mu = 0$
\item $A x^* \leq b$
\item $\mu \geq 0$
\item $\mu^\top (b - Ax) = 0$
\end{enumerate}
\end{definition}

We can now continue with our proof of \cref{thm:membership}.
The problem of finding an exact KKT point of \eqref{eqn:equivminformulation} lies in \CLS. Indeed, it is known that finding an approximate KKT point of such a program lies in \CLS~\cite[Theorem~5.1]{fearnley2022complexity}, and, given that the objective function is a quadratic polynomial, an approximate KKT point (with sufficiently small approximation error) can be turned into an exact one in polynomial time (see, e.g., \cite[Lemma~A.1]{Fearnley23}).

\cref{thm:membership} now simply follows from the following claim.

\begin{claim}\label{clm:KKTNash}
Given a KKT point of \eqref{eqn:equivminformulation}, we can compute a Nash equilibrium of the original game in polynomial time.
\end{claim}

\begin{proof}
Let $(x^*,\gamma^*)$ be a KKT point of \eqref{eqn:equivminformulation}. We define notation for the following multipliers:
\begin{itemize}
\item For all $j \in Y$ and $k \in [m]$, $\mu_{jk}^* \in \mathbb{R}_{\geq 0}$ corresponding to the constraint $\gamma_j \geq \sum_{i \in X} e_k^\top A^{j,i} x_i$.
\item For all $i \in X$, $\lambda_i^* \in \mathbb{R}$ corresponding to the constraint $\sum_{k \in [m]} x_{ik} = 1$.
\item For all $i \in X$ and $k \in [m]$, $\nu_{ik}^* \in \mathbb{R}_{\geq 0}$ corresponding to the constraint $x_{ik} \geq 0$.
\end{itemize}
Note that we have not included multipliers for the constraints $\gamma_j \leq M$. This is because none of these constraints will ever be tight at a KKT point $(x^*,\gamma^*)$ by construction of $M$.

Now, since $(x^*,\gamma^*)$ is a KKT point, we can compute such multipliers that satisfy the following KKT conditions (which are derived from \cref{def:min-KKT-general}) in polynomial time\footnote{Indeed, it is easy to check that given $(x^*,\gamma^*)$ such multipliers can then be found by solving an LP.}:
\begin{enumerate}
\item For all $i \in X$
\[- \sum_{i' \neq i} A^{i,i'}x_{i'}^* - \sum_{j \in Y} A^{i,j} \mu_j^* + \lambda_i^* \cdot 1_m - \nu_i^* = 0\]
where $1_m$ denotes a vector of $m$ ones, and where we used the fact that $A^{i,j} = -(A^{j,i})^\top$. Additionally, we have that for all $k \in [m]$, $\nu_{ik}^* > 0 \implies x_{ik}^* = 0$.
\item For all $j \in Y$
\[1 - \sum_{k \in [m]} \mu_{jk}^* = 0.\]
\item For all $j \in Y$ and $k \in [m]$, $\mu_{jk}^* > 0 \implies \gamma_j^* = \sum_{i \in X} e_k A^{j,i} x_i^*$.
\end{enumerate}

Now, we claim that $(x^*,\mu^*)$ forms a Nash equilibrium of the original game. First of all, note that by property 2 we have $\sum_{k \in [m]} \mu_{jk}^* = 1$ and thus $\mu_j^*$ is indeed a valid mixed strategy for player $j \in Y$. Next, property 3 can be rewritten as
\[\mu_{jk}^* > 0 \implies U_j(x^*,\mu_{-j}^*,e_k) = \sum_{i \in X} e_k A^{j,i} x_i^* = \gamma_j^* \geq \max_{k' \in [m]} \sum_{i \in X} e_{k'} A^{j,i} x_i^* = \max_{k' \in [m]} U_j(x^*,\mu_{-j}^*,e_{k'})\]
which means that the strategy of player $j$, namely $\mu_j^*$, is a best-response. Finally, property 1 can be reinterpreted as saying that $x_i^*$ is a KKT point of the following optimization problem (in variables $x_i$)
\begin{equation*}
\begin{split}
\min \quad & - \sum_{i' \neq i} x_i A^{i,i'} x_{i'}^* - \sum_{j \in Y} x_i A^{i,j} \mu_j^* \quad [= U_i(x_i, x_{-i}^*,\mu^*)] \\
\text{s.t. } & x_i \in \mathcal{X}_i
\end{split}
\end{equation*}
Since this is an LP, any KKT point is also a global solution. Thus, $x_i^*$ is a global minimum of this LP, which means that player $i$ is also playing a best-response.
\end{proof}

\subsection{Algorithm for Computing an Approximate Nash Equilibrium}

To provide an algorithm and some theoretical guarantees, we have to characterize an $\varepsilon$-KKT point for the minimization problem from the previous section, with the following definition:

\begin{definition}\label{def:min-KKT-general-eps}
Consider an optimization problem of the form
\begin{equation*}
\begin{split}
\min \quad &f(x) \\
\text{s.t. } &Ax \leq b \\
&x \in \mathbb{R}^n
\end{split}
\end{equation*}
where $f: \mathbb{R}^n \to \mathbb{R}$ is continuously differentiable, $A \in \mathbb{R}^{m \times n}$, and $b \in \mathbb{R}^m$. A point $\tilde{x} \in \mathbb{R}^n$ is a $\varepsilon$-KKT point of this problem if there exists $\mu \in \mathbb{R}^m$ such that
\begin{enumerate}
\item $\lVert \nabla f(\tilde{x}) + \mu^\top A\rVert_2 \leq \varepsilon $
\item $A \tilde{x} \leq b$
\item $\mu \geq 0$
\item $\mu^\top (A\tilde{x} - b) = 0$
\end{enumerate}
\end{definition}

As is standard in constrained optimization, it will be useful to define the \emph{Normal cone} to a nonempty closed convex feasible set $\mathcal {X} \subseteq \mathbb{R}^n$ at a point $x \in \cal X$.
\begin{definition}\label{def:normal-cone}
    $N_{\cal X}(x):=\{d \in \mathbb{R}^n:d^\top (y-x) \leq 0,\;\text{for all}\; y \in \cal{X}\}$.
\end{definition}

It is well known that if $x^*$ is an exact stationary point, then a necessary condition is that $-\nabla f(x^*) \in N_{\cal X}(x^*)$. Similarly, if $\tilde{x}$ is a $\varepsilon$-KKT point, then such a point can be characterized in terms of saying $\text{dist}(-\nabla f(\tilde{x},N_{\cal X}(\tilde{x}))\leq \varepsilon$, or equivalently, $-\nabla f(x^*) \in N_{\cal X}(x^*)+\mathbb{B}(0,\varepsilon)$, where $\mathbb{B}(0,\varepsilon)$ is the Euclidean ball centered at 0 and of radius $\varepsilon$, and the ``+'' operator here refers to the Minkowski sum of the two sets.

Furthermore, in our case, the normal cone can be characterized quite easily. In the exact KKT point case, it boils down to being able to represent the negative gradient vector as a conic combination of vectors in the normal cone. By extension, the approximate version, states that negative gradient at that point needs to be at most $\varepsilon$ away in Euclidean distance to a point in the normal cone. 

\begin{theorem}
Assume all payoffs lie in $[-1,1]$.
There exists an algorithm that given $\varepsilon > 0$ and a two-team zero-sum polymatrix game with independent adversaries $\Gamma$ outputs an $\varepsilon$-Nash equilibrium in time
\[
O\!\left(\frac{\mathrm{poly}(|\Gamma|)}{\varepsilon^2}\right)
\]
where $|\Gamma|$ is the size of the representation of the game, i.e., the representation of the payoff matrices.
\end{theorem}

\begin{proof}
We first rely on the reformulation of our minmax problem into an equivalent minimization problem from \cref{thm:membership}. Particularly we are going to work with the formulation described in \eqref{eqn:equivminformulation}.
Since our constraints are bounded polytopes, we can always ensure \emph{exact} feasibility and will enforce \emph{exact} complementary slackness.
If $(\tilde{x},\tilde{\gamma})$ is a $\varepsilon$-KKT point, we describe the guarantees to compute the multipliers exactly that satisfy the following KKT conditions (which are derived from \cref{def:min-KKT-general-eps}). We describe below how the properties appear when characterizing an $\varepsilon$-KKT point (cf.~\cref{clm:KKTNash}).

\begin{enumerate}
\item For all $i \in X$
\[\Big\lVert - \sum_{i' \neq i} A^{i,i'}\tilde{x}_{i'} + \sum_{j \in Y} \tilde{\mu}_jA^{i,j}+ \tilde{\lambda_i}1_m - \tilde{\nu_i}\Big\rVert_2 \leq \varepsilon\]
where $1_m$ denotes a vector of $m$ ones. Additionally, we have that for all $k \in [m]$, $\tilde{\nu}_{ik} > 0 \implies \tilde{x}_{ik} = 0$.
\item For all $j \in Y$
\[\Big \lvert 1 - \sum_{k \in [m]} \tilde{\mu}_{jk} \Big\rvert \leq \varepsilon.\]
\item For all $j \in Y$ and $k \in [m]$, $\tilde{\mu}_{jk} > 0 \implies \tilde{\gamma}_j = \sum_{i \in X} e_k A^{j,i} \tilde{x}_i$.
\end{enumerate}

We may define $\hat{\mu}_{jk}=\dfrac{\tilde{\mu}_{jk}}{\sum_{\ell \in [m]}\tilde{\mu}_{j\ell}}$. We will now show that $(\tilde{x},\hat{\mu})$ form an $O(\mathrm{poly}(|\Gamma|)\varepsilon)$-NE for the original problem.

Firstly, consider $\hat{\mu}_{jk} > 0$, this implies $\tilde{\mu}_{jk} > 0$, which further implies that \[U_j(\tilde{x},\hat{\mu}_{-j},e_k) = \sum_{i \in X} e_k A^{j,i} \tilde{x}_i = \tilde{\gamma}_j \geq \max_{k' \in [m]} \sum_{i \in X} e_{k'} A^{j,i} \tilde{x}_i = \max_{k' \in [m]} U_j(\tilde{x},\hat{\mu}_{-j},e_{k'}),\]
which means that the strategy of player $j$, namely $\hat{\mu}_j$, is a best-response.

Consider, the term \[\Big\lVert - \sum_{i' \neq i} A^{i,i'}\tilde{x}_{i'} + \sum_{j \in Y} \hat{\mu}_jA^{i,j}+ \tilde{\lambda_i}1_m - \tilde{\nu_i}\Big\rVert_2.\] 
By substituting the definition of $\hat{\mu}$, we can show that the above expression is upper bounded by: 
\[\Big\lVert - \sum_{i' \neq i} A^{i,i'}\tilde{x}_{i'} + \sum_{j \in Y} \tilde{\mu}_jA^{i,j}+ \tilde{\lambda_i}1_m - \tilde{\nu_i}\Big\rVert_2+ \Bigg(\dfrac{\lvert 1-\sum_{\ell \in [m]}\tilde{\mu}_{j\ell}\rvert}{\lvert\sum_{\ell \in [m]}\tilde{\mu}_{j\ell}\rvert} \Bigg)\Big\lVert\sum_{j \in Y} \tilde{\mu}_jA^{i,j}\Big\rVert_2 \] 

Now, we use condition 2 from above and the fact that payoffs are bounded between [-1,1] to get a further upper bound as follows:
\[\Big\lVert - \sum_{i' \neq i} A^{i,i'}\tilde{x}_{i'} + \sum_{j \in Y} \tilde{\mu}_jA^{i,j}+ \tilde{\lambda_i}1_m - \tilde{\nu_i}\Big\rVert_2+ \Bigg(\dfrac{\varepsilon}{1-\varepsilon}\Bigg)|Y|(1+\varepsilon)m^{3/2}.\]

Furthermore, from condition 1 and assuming without loss of generality that $\varepsilon < 1/2$, we get that the above term is overall bounded by:
\begin{align}\label{eqn:eps-KKT-cond1}
    \Big\lVert -\sum_{i' \neq i} A^{i,i'}\tilde{x}_{i'} + \sum_{j \in Y} \hat{\mu}_jA^{i,j}+ \tilde{\lambda_i}1_m - \tilde{\nu_i}\Big \rVert_2 \leq O(\mathrm{poly}(|\Gamma|)\varepsilon)
\end{align}

Now consider, as in the exact case, 

\begin{equation}\label{eqn:x-best-response}
\begin{split}
\min \quad & - \sum_{i' \neq i} x_i A^{i,i'} \tilde{x}_{i'} + \sum_{j \in Y} \hat{\mu}_jA^{i,j}x_i \quad [= U_i(x_i, \tilde{x}_{-i},\hat{\mu})] \\
\text{s.t. } & x_i \in \mathcal{X}_i
\end{split}
\end{equation}

From \eqref{eqn:eps-KKT-cond1}, $\tilde{x}_i$ can be seen as an $O(\mathrm{poly}(|\Gamma|)\varepsilon)$-KKT point of the above problem. However, since the objective function is an LP with a feasible set that is bounded in diameter, it can be shown that a $\delta$-KKT point becomes a $\sqrt{2}\delta$-approximate global minimum for the above problem. 
\begin{claim}
    A $\delta$-KKT point for the optimization problem \eqref{eqn:x-best-response} is a $\sqrt{2}\delta$-approximate global minimum for the problem.
\end{claim}

\begin{proof}
This can be seen by the following argument:
First define $h(x)=U_i(x,\tilde{x}_{-i},\hat{\mu})$. Due to the convexity of $h$, we have:
$h(y) \geq h(x)+\nabla h(x)^\top (y-x)$. Let $x_i^*$ be a global minimizer of $h$, then by the above inequality, we have: 
$h(x_i^*) \geq h(\tilde{x}_i)+\nabla h(\tilde{x}_i)^\top (x_i^*-\tilde{x}_i)$. Consider a vector $d \in N_{\mathcal{X}_i}(\tilde{x}_i)$, which implies $d^\top (x_i^*-\tilde{x}_i) \leq 0$. Applying this to the above inequality, we have:
$h(x_i^*) \geq h(\tilde{x}_i)+\nabla h(\tilde{x}_i)^\top (x_i^*-\tilde{x}_i) \geq (\nabla h(\tilde{x}_i)+d)^\top (x_i^*-\tilde{x}_i)$. 
Now a lower bound on $(\nabla h(\tilde{x}_i)+d)^\top (x_i^*-\tilde{x}_i)$ is simply $-\lVert (\nabla h(\tilde{x}_i)+d)\rVert_2 \lVert(x_i^*-\tilde{x}_i) \rVert_2$.
Since $\tilde{x}_i$ is a $\delta$-KKT point, there exists $d \in N_{\mathcal{X}_i}(\tilde{x}_i)$ with $\lVert (\nabla h(\tilde{x}_i)+d)\rVert_2 \leq \delta$.
Finally, both $x_i^*$ and $\tilde{x}_i$ are points in the same probability simplex, whose diameter is upper bounded by $\sqrt{2}$. Thus, we obtain $h(\tilde{x}_i) \leq h(x_i^*) + \sqrt{2}\delta$.
\end{proof}
Applying this claim for each player $i$, we get that 
 $(\tilde{x},\hat{\mu})$ is a $O(\mathrm{poly}(|\Gamma|)\varepsilon)$-NE for the original problem.
\paragraph{Computing Multipliers:}
Let $J_j$ be the index set of the inequality constraints at $\tilde{x}$ for player $j$
which corresponds to the set of inequalities $\tilde{\gamma}_j \geq \sum_{i \in X}e_k^\top A^{j,i}\tilde{x}_i$, for all $k \in [m]$. Let $\tilde{J}_j$ be the set of active constraints where the above inequalities attain equality. Then let us set $\mu_{jk}=0$ for $k \in J_j \backslash \tilde{J}_j$. Now let $I_i$ be the index set for the inequality constraints concerning $\tilde{x}_i$, and let $\tilde{I}_i$ denote the set of active constraints for $\tilde{x}_i$, then we set $\nu_{ik}=0$ for each $k \in I_i \backslash \tilde{I}_i$ . For each player $i$ in the team, consider the following optimization problem:

\begin{equation*}\label{eqn:compute_mul}
\begin{split}
\min \quad &\Big\lVert - \sum_{i' \neq i} A^{i,i'}\tilde{x}_{i'} + \sum_{j \in Y}\sum_{k \in \tilde{J}_j} \mu_{jk}e_k^\top A^{i,j}+ \lambda_i1_m - \nu'_i\Big\rVert^2\\
\text{s.t. } &\Big\lvert \sum_{k\in \tilde{J}_j}\mu_{jk} - 1\Big\rvert \leq \varepsilon\;\;\text{for all}\;\; j \in Y\\
& \mu_{jk} \geq 0\;\; \text{for all}\;\; j \in Y\;\;\text{and}\;\;k\in \tilde{J}_j\\
& \nu'_{ik} \geq 0\;\; \text{for all}\;\;k\in \tilde{I}_i\\
& \nu'_{ik} = 0\;\; \text{for all}\;\;k\in I_i \backslash \tilde{I}_i\\
& \lambda_i \in \mathbb{R}
\end{split}
\end{equation*}
Any optimal solution to the above program (obtained for each player $i \in X$) gives us the required $\tilde{\mu},\tilde{\lambda},\tilde{\nu}$, that satisfy complementary slackness and the $\varepsilon$-KKT conditions. Note that the above program is a convex quadratic program with a bounded feasible set that is a polytope. Hence, the multipliers can be found \emph{exactly} in time that is polynomial in $|\Gamma|$, see, e.g., \cite{KozlovTK80-convex-quadratic}.

\paragraph{Algorithm:}
 We will run projected gradient descent (PGD) on our original minimization problem that is \eqref{eqn:equivminformulation} and let us call the feasible set in this program to be $\Theta$. In our case projections can be obtained efficiently. Suppose, we get some $(x,\gamma)$, which is infeasible. First for every infeasible $x_i$, we project onto the respective simplex and then for each $j \in Y$ we choose $\gamma_j = \max_{k \in m} \sum_{i \in X}e_k^\top A^{j,i}x_i$. It is easy to see from this procedure that per step projection takes time that is polynomial in $|\Gamma|$. Now, the PGD\footnote{Here we simply use $\nabla f(x^t,\gamma
^t):=[\nabla_xf(x^t,\gamma^t),\nabla_\gamma f(x^t,\gamma^t)]$.} update can be written as:
\begin{align}\label{eqn:PGD}
    (x^{t+1},\gamma^{t+1})=P_{\Theta}\left((x^t,\gamma^t)-\frac{1}{L}\nabla f(x^t,\gamma^t)\right)
\end{align}

One can define the gradient mapping:
\begin{align}\label{eqn:GMP}
    G_L(x,\gamma)=L\left((x,\gamma)-P_{\Theta}\left((x,\gamma)-\frac{1}{L}\nabla f(x,\gamma)\right)\right)
\end{align}
It is well known that under these conditions PGD has the following guarantee:

\begin{lemma}[Section 7.3.1 in \cite{wright2022optimization}]
Let $\mathcal X\subset \mathbb R^n$ be a nonempty closed convex set and let
$f:\mathbb{R}^n \rightarrow \mathbb{R}$ be $L$-smooth and bounded from below on $\mathcal X$,
i.e., $f^* := \inf_{x \in\mathcal X} f(x) > -\infty$.
Fix $\varepsilon > 0$ and $x^0 \in \mathcal X$, and generate $\{x^t\}_{t\ge 0}$ by projected
gradient descent with stepsize $1/L$.
If
\[
T \;\ge\; \frac{2L\bigl(f(x^0)-f^*\bigr)}{\varepsilon^2},
\]
then:
\begin{enumerate}
    \item $\displaystyle \min_{0\leq t\leq T-1}\lVert G_L(x^t)\rVert \leq \varepsilon/2$.
    \item There exists $t \in \{1,2,\ldots,T\}$ such that
    \[
    -\nabla f(x^t) \in N_{\mathcal X}(x^t) + \mathbb{B}(0,\varepsilon).
    \]
\end{enumerate}
Here $N_{\mathcal X}(x)$ denotes the normal cone to $\mathcal X$ at $x$, and
$\mathbb{B}(0,\varepsilon)$ is the Euclidean ball of radius $\varepsilon$ centered at $0$.
\end{lemma}

Thus, we run PGD until we encounter an iterate $t$ such that $\|G_L(x^t)\|\le \varepsilon/2$
(the lemma guarantees that such a $t$ exists within
$T = O\!\left(\frac{L(f(x^0)-f^*)}{\varepsilon^2}\right)$ iterations), and then set
$(\tilde{x},\tilde{\gamma}) = (x^{t+1},\gamma^{t+1})$.
By the projection optimality condition and $L$-smoothness, this implies that
$-\nabla f(\tilde{x},\tilde{\gamma}) \in N_{\mathcal X}(\tilde{x},\tilde{\gamma}) + \mathbb{B}(0,\varepsilon)$.
To get an $O(\varepsilon)$-NE for our problem, we run PGD to obtain an $\frac{\varepsilon}{\mathrm{poly}(|\Gamma|)}$-KKT point. We then solve the convex program in \eqref{eqn:compute_mul} at $(\tilde{x},\tilde{\gamma})$ . Since our payoffs are bounded in $[-1,1]$, our objective function is  Lipschitz continuous and Lipschitz smooth with constants that are polynomially bounded in $|\Gamma|$. Furthermore, the per step running time (gradient computation and the projection operation) is polynomially bounded in $|\Gamma|$. Thus, the final running time is
$O\!\left(\frac{\mathrm{poly}(|\Gamma|)}{\varepsilon^2}\right)$.
\end{proof}

\section{Open Problems}

We believe the following are some interesting open questions for future work:
\begin{itemize}
    \item Our \CLS-membership result only applies to the setting where the adversaries are independent. If interactions between adversaries are allowed then the problem is only known to lie in \PPAD, so there is a gap with the \CLS-hardness that we show. Is the problem \CLS-complete, \PPAD-complete, or neither?
    \item Our hardness result provides strong evidence that no algorithm with running time $O(\poly(\allowbreak\log(\allowbreak 1/\eps)))$ exists. However, we show that PGD attains a running time (with respect to $\varepsilon$) of $O(1/\eps^2)$, for the polymatrix setting. What is the optimal polynomial dependence in $1/\eps$ in this setting?
\end{itemize}

\subsection*{Acknowledgements}

AH and GM were supported by the Swiss State Secretariat for Education, Research and Innovation (SERI) under contract number MB22.00026. SGN would like to thank Ioannis Panageas and colleagues at the Zuse Institute Berlin (ZIB) for discussions pertaining to this project.

\small 
\DeclareUrlCommand{\Doi}{\urlstyle{sf}} 
\renewcommand{\path}[1]{\small\Doi{#1}} 
\renewcommand{\url}[1]{\href{#1}{\small\Doi{#1}}} 
\bibliographystyle{alphaurl} 
\bibliography{bibliography.bib}
\end{document}